\RequirePackage{amsmath}
\documentclass[runningheads,a4paper]{llncs}
\usepackage[utf8]{inputenc}
\usepackage[ruled,vlined,linesnumbered]{algorithm2e}
\usepackage{amsmath}
\usepackage{amsfonts}
\usepackage{amssymb}
\usepackage{fix-cm}
\usepackage{microtype}
\usepackage{graphicx}
\usepackage{authblk}
\usepackage{xspace}

\newcommand{\abseas}{\texttt{ABSE-AS}\xspace}
\newcommand{\absems}{\texttt{ABSE-MS}\xspace}

\begin{document}

\title{Effective Edge-Fault-Tolerant Single-Source Spanners via Best (or Good) Swap Edges}
\author{Davide~Bil\`o\inst{1} \and Feliciano~Colella\inst{2} \and Luciano~Gual\`a\inst{3} \and Stefano~Leucci\inst{4} \and Guido~Proietti\inst{5,6}}
\titlerunning{Edge-Fault-Tolerant Single-Source Spanners via Best (or Good) Swap Edges}
\authorrunning{Davide Bil\`o et al.}

\institute{
Dipartimento di Scienze Umanistiche e Sociali, University of Sassari, Italy \email{davide.bilo@uniss.it}\and
Gran Sasso Science Institute, L'Aquila,  Italy. \email{feliciano.colella@gssi.it}\and
Dipartimento di Ingegneria dell'Impresa, University of Rome ``Tor Vergata'',  Italy. \email{guala@mat.uniroma2.it}\and
Department of Computer Science, ETH Zürich, Switzerland. \email{stefano.leucci@inf.ethz.ch}\and
Dipartimento di Ingegneria e Scienze dell'Informazione e Matematica, University of L'Aquila, Italy. \email{guido.proietti@univaq.it}\and
Istituto di Analisi dei Sistemi ed Informatica, CNR, Rome, Italy.
}
\maketitle

\begin{abstract}
Computing \emph{all best swap edges} (ABSE) of a spanning tree $T$ of a given $n$-vertex and $m$-edge undirected and weighted graph $G$ means to select, for each edge $e$ of $T$, a corresponding non-tree edge $f$, in such a way that the tree obtained by replacing $e$ with $f$ enjoys some optimality criterion (which is naturally defined according to some objective function originally addressed by $T$).
Solving efficiently an ABSE problem is by now a classic algorithmic issue, since it conveys a very successful way of coping with a (transient) \emph{edge failure} in tree-based communication networks: just replace the failing edge with its respective swap edge, so as that the connectivity is promptly reestablished by minimizing the rerouting and set-up costs.
In this paper, we solve the ABSE problem for the case in which $T$ is a \emph{single-source shortest-path tree} of $G$, and our two selected swap criteria aim to minimize either the \emph{maximum} or the \emph{average stretch} in the swap tree of all the paths emanating from the source. Having these criteria in mind, the obtained structures can then be reviewed as \emph{edge-fault-tolerant single-source spanners}. For them, we propose two efficient algorithms running in $O(m n  +n^2 \log n)$ and $O(m n \log \alpha(m,n))$ time, respectively, and we show that the guaranteed (either maximum or average, respectively) stretch factor is equal to 3, and this is tight. Moreover, for the maximum stretch, we also propose an almost linear $O(m \log \alpha(m,n))$ time algorithm computing a set of \emph{good} swap edges, each of which will guarantee a relative approximation factor on the maximum stretch of $3/2$ (tight) as opposed to that provided by the corresponding BSE. Surprisingly, no previous results were known for these two very natural swap problems.
\end{abstract}

\section{Introduction}
Nowadays there is an increasing demand for an \emph{efficient} and \emph{resilient} information exchange in communication networks.
This means to design on one hand a logical structure onto a given communication infrastructure, which optimizes some sought routing protocol in the absence of failures, and on the other hand, to make such a structure resistant against possible link/node malfunctioning, by suitably adding to it a set of redundant links, which will enter into operation as soon as a failure takes place.

More formally, the depicted situation can be modeled as follows: the underlying communication network is an $n$-vertex and $m$-edge undirected input graph $G=(V(G),E(G),w)$, with positive real edge weights defined by $w$, the logical (or primary) structure is a (spanning) subgraph $H$ of $G$, and finally the additional links is a set of edges $A$ in $E(G) \setminus E(H)$. Under normal circumstances, communication takes place on $H$, by following a certain protocol, but as soon as an edge in $H$ fails, then one or more edges in $A$ come into play, and the communication protocol is suitably adjusted.

In particular, if the primary structure is a (spanning) \emph{tree} of $G$, then a very effective way of defining the set of additional edges is the following: with each tree edge, say $e$, we associate a so-called \emph{best swap edge}, namely a non-tree edge that will replace $e$ once it (transiently) fails, in such a way that the resulting \emph{swap tree} enjoys some nice property in terms of the currently implemented communication protocol. By doing in this way, rerouting and set-up costs will be minimized, in general, and the quality of the post-failure service remains guaranteed. Then, an \emph{all best swap edges} (ABSE) problem is that of finding efficiently (in term of time complexity) a best swap edge for each tree edge.

Due to their fault-tolerance application counterpart, ABSE problems received a large attention by the algorithmic community.
In such a framework, a key role has been  played by the \emph{Shortest-Path Tree} (SPT) structure, which is commonly used for implementing efficiently the \emph{broadcasting} communication primitive.
Indeed, it is was shown already in \cite{IIOY05} that an effective post-swap broadcast protocol can be put in place just after the original SPT undergoes an edge failure. Not surprisingly then, several ABSE problems w.r.t. an SPT have been studied in the literature, for many different swap criteria.

\paragraph{Previous work on swapping in an SPT.}
Since an SPT enjoys several optimality criteria when looking at distances from the source, say $s$, several papers have analyzed the problem in various respects. However, most of the efforts focused on the minimization w.r.t. the following two swap criterion: the maximum/average distance from $s$ to any node which remained disconnected from $s$ after a failure. The currently fastest solutions for these two ABSE problems run in $O(m \log \alpha(m,n))$ time \cite{BGP15} and $O(m \, \alpha(n,n) \log^2 n)$ time \cite{DP07}, respectively.
Moreover, it has been shown that in the swap tree the maximum (resp., average) distance of the disconnected nodes from $s$ is at most twice (resp., triple) that of the new optimum SPT~\cite{NPW03}, and these bounds are tight.

Other interesting swap criteria which have been analyzed include the minimization of the maximum increase (before and after the failure) of the distance from $s$, and the minimization of the distance from $s$ to the root of the subtree that gets disconnected after the failure \cite{NPW01}. Besides the centralized setting, all these swap problems have been studied also in a distributed framework (e.g., see \cite{FEPPS04,FEPPS06,FPPSW08,DLPP13}).

On the other hand, no results are known for the case in which one is willing to select a BSE with the goal of minimizing either the \emph{maximum} or the \emph{average stretch} from the source $s$ of the disconnected nodes, where the stretch of a node is measured as the ratio between its distance from $s$ in the swap tree and in a new optimum SPT.
This is very surprising, since they are (especially the former one) the universally accepted criterion leading to the design of a
\emph{spanner}, i.e., a sparse subgraph preserving shortest paths (between pairs of vertices of interest) in a graph (also in the presence of failures).

In this paper, we aim to fill this gap, by providing efficient solutions exactly for these two swap criteria.

\paragraph{Our results.}
Let us denote by \absems and \abseas the ABSE problem w.r.t. the maximum and the average stretch swap criterion, respectively.
For such problems, we devise two efficient algorithms running in $O(m n +n^2 \log n)$ and $O(m n \log \alpha(m,n))$ time, respectively. Notice that both solutions incorporate the running time for computing all the replacement shortest paths from the source after the failure of every edge of the SPT, as provided in \cite{GP07}, whose computation essentially dominates in an asymptotic sense the time complexity.
Our two solutions are based on independent ideas, as described in the following:

\begin{itemize}
\item for the \absems problem, we develop a \emph{centroid decomposition} of the SPT, and we exploit a distance property that has to be enjoyed by a BSE w.r.t. a nested and log-depth hierarchy of centroids, which will be defined by the subtree detached from the source after the currently analyzed edge failure. A further simple filtering trick on the set of potential swap edges will allow to reduce them from $O(m)$ to $O(n)$, thus returning the promised $O(n^2 \log n)$ time.
\item for the \abseas problem, we instead suitably combine a set of linearly-computable (at every edge fault) information, that essentially will allow to describe in $O(1)$ time the quality of a swap edge. This procedure is in principle not obvious, since to compute the average stretch we need to know, for each swap edge, the $O(n)$ distances to all the nodes in the detached subtree.  Again, by filtering on the set of potential swap edges, we will get an $O(n^2)$ running time, which will be absorbed by the all-replacement paths time complexity.
\end{itemize}

\noindent
Concerning the quality of the corresponding swap trees,  we instead show that the guaranteed
(either maximum or average, respectively) stretch factor w.r.t. the paths emanating from the source (in the surviving graph) is equal to 3, and this is tight.
By using a different terminology, our structures can then be revised as \emph{edge-fault-tolerant single-source $3$-spanners}, and we qualified them as \emph{effective} since they can be computed quickly, are very sparse, provide a very simple alternative post-failure routing, and finally have a small (either maximum or average) stretch.

Although the proposed solutions are quite efficient, their running time can become prohibitive for large and dense input graphs, since in this case they would amount to a time cubic in the number of vertices. Unfortunately, it turns out that their improvement is unlikely to be achieved,
unless one could avoid the explicit recomputation of all post-failure distances from the source.
To circumvent this problem, we then adopt a different approach, which by the way finds application for the (most relevant) max-stretch measure only: we renounce to optimality in the detection of a BSE, in return of a substantial improvement (in the order of a linear factor in $n$) in the runtime.
More precisely, for such a measure, we will compute in an almost linear $O(m \log \alpha(m,n))$ time a set of \emph{good} swap edges (GSE), each of which will guarantee a relative approximation factor on the maximum stretch of $3/2$ (tight) as opposed to that provided by the corresponding BSE. Moreover, a GSE will still guarantee an absolute maximum stretch factor w.r.t. the paths emanating from the source (in the surviving graph) equal to 3 (tight).

Besides that, we also point out another important feature concerned with the computation in a \emph{distributed} setting of all our good swap edges. Indeed, in \cite{DLPP13} it was shown that they can be computed in an \emph{asynchronous message passing system} in essentially optimal \emph{ideal time},\footnote{This is the time obtained with the ideal assumption that the communication time of each message to a neighboring process takes constant time, as in the synchronous model.} space usage, and message complexity, as opposed to the recomputation of all the corresponding BSE, for which
no efficient solution is currently available.

\paragraph{Other related results.}
Besides swap-based approaches, an SPT can be made edge-fault-tolerant by further enriching the set of additional edges, so that the obtained structure has almost-shortest paths emanating from the source, once an edge fails. The currently best trade off between the size of the set of additional edges and the quality of the resulting paths emanating from $s$ is provided in \cite{BGLP14}, where the authors showed that for any arbitrary constant $\varepsilon>0$, one can compute in polynomial time a slightly superlinear (in $n$, and depending on $\varepsilon$) number of additional edges in such a way that the resulting structure retains $(1+\varepsilon)$-stretched post-failure paths from the source.

For the sake of completeness, we also quickly recall the main results concerned with ABSE problems. For the \emph{minimum spanning tree} (MST), a BSE is of course one minimizing the \emph{cost} of the swap tree, i.e., a swap edge of minimum cost. This problem is also known as the MST \emph{sensitivity analysis} problem, and can be solved in $O(m\log\alpha(m,n))$ time~\cite{Pet05}.
Concerning the \emph{minimum diameter spanning tree}, a BSE is instead one minimizing the \emph{diameter} of the swap tree~\cite{IR98,NPW01}, and the best solution runs in $O(m \log \alpha(m,n))$ time \cite{BGP15}. Regarding the \emph{minimum routing-cost spanning tree}, a BSE is clearly one minimizing the \emph{all-to-all routing cost} of the swap tree~\cite{WHC08}, and the fastest solutions for solving this problem has a running time of $O\left(m 2^{O(\alpha(n,n))}\log^2 n\right)$~\cite{BGP14}. Finally, for a \emph{tree spanner}, a BSE is one minimizing the maximum stretch w.r.t. the all pair distances, and the fastest solution to date run in $O(m^2 \log \alpha(m,n))$ time \cite{BiloCGP15}.

To conclude, we point out that the general problem of designing fault-tolerant spanners for the \emph{all-to-all} case has been extensively studied in the literature, and we refer the interested reader to \cite{ChechikLPR09,DinitzK11,BiloGG0P15} and the references therein.

\section{Problem Definition}
Let $G = (V(G), E(G), w)$ be a $2$-edge-connected, edge-weighted, and undirected graph with cost function $w : E(G) \rightarrow \mathbb{R}^+$. We denote by $n$ and $m$ the number of vertices and edges of $G$, respectively. If $X \subseteq V$, let $E(X)$ be the set of edges incident to at least one vertex in $X$.
Given an edge $e \in E(G)$, we will denote by $G-e$ the graph obtained from $G$ by removing edge $e$. Similarly, given a vertex $v \in V(G)$, we will denote by $G-v$ the graph obtained from $G$ by removing vertex $v$ and all its incident edges. Let $T$ be an SPT of $G$ rooted at $s \in V(G)$. Given an edge $e \in E(T)$, we let $C(e)$ be the set of all the \emph{swap edges} for $e$, i.e., all edges in $E(G) \setminus \{ e \}$ whose endpoints lie in two different connected components of $T-e$, and let $C(e,X)$ be the set of all the swap edge for $e$ incident to a vertex in $X \subseteq V(G)$. For any $e \in E(T)$ and $f \in C(e)$, let $T_{e/f}$ denote the \textit{swap tree} obtained from $T$ by replacing $e$ with $f$. Let $T_v = (V(T_v), E(T_v))$ be the subtree of $T$ rooted at $v \in V(G)$.
Given a pair of vertices $u,v \in V(G)$, we denote by $d_G(u,v)$ the \emph{distance} between $u$ and $v$ in $G$. Moreover, for a swap edge $f=(x,y)$, we assume that the first appearing endvertex is the one closest to the source, and we may denote by $w(x,y)$ its weight. We define the \textit{stretch factor of $y$ w.r.t. $s,T,G$} as $\sigma_G(T,y) = \frac{d_T(s,y)}{d_G(s,y)}$.

Given an SPT $T$ of $G$, the \absems problem is that of finding, for each edge $e=(a,b) \in E(T)$, a swap edge $f^*$ such that:

\[
	f^* \in \arg \min_{f \in C(e)} \left\{\mu(f):= \max_{v\in V(T_b)}\sigma_{G-e}\big(T_{e/f},v\big) \right\}.
\]

\noindent Similarly, the \abseas problem is that of finding, for each edge $e=(a,b) \in E(T)$, a swap edge $f^*$ such that:
\[
	f^* \in \arg \min_{f \in C(e)} \Bigg\{ \lambda(f):=\frac{1}{|V(T_b)|}\sum_{v\in V(T_b)}  \sigma_{G-e}\big(T_{e/f},v\big) \Bigg\}.
\]

\noindent
We will call $\mu(f)$ (resp., $\lambda(f)$) the \emph{max-}(resp., \emph{avg-})\emph{stretch} of $f$ w.r.t. $e$.

\section{An Algorithm for \absems}\label{sec:absemsalg}
In this section we will show an efficient algorithm to solve the \absems problem in $O(mn + n^2 \log n)$ time.
Notice that a brute-force approach would require $O(mn^2)$ time, given by the $O(n)$ time which is needed to evaluate the quality of each of the $O(m)$ swap edges, for each of the $n-1$ edges of $T$. Our algorithm will run through $n-1$ phases, each returning in $O(m+n\log n)$ time a BSE for a failing edge of $T$, as described in the following.

Let us fix $e=(a,b)$ as the failing edge. First, we compute in $O(m+ n \log n)$ time all the distances in $G-e$ from $s$. Then, we filter the $O(m)$ potential swap edges to $O(n)$, i.e., at most one for each node $v$ in $T_b$. Such a filtering is simply obtained by selecting, out of all edges $f=(x,v) \in C(e, \{ v  \})$, the one minimizing the measure $d_G(s,x)+w(f)$. Indeed, it is easy to see that the max-stretch of such selected swap edge is never worse than that of every other swap edge in $C(e)$. This filtering phase will cost $O(m)$ total time. As a consequence, we will henceforth assume that $|C(e)| = O(n)$.

Then, out of the obtained $O(n)$ swap edges for $e$, we further restrict our attention to a subset of $O(\log n)$ \emph{candidates} as BSE, which are computed as follows. Let $\Lambda$ denote a generic subtree of $T_b$, and assume that initially $\Lambda=T_b$.
First of all, we compute in $O(|V(\Lambda)|)$ time a \emph{centroid} $c$ of $\Lambda$, namely a node whose removal from $\Lambda$ splits $\Lambda$ in a forest $F$ of subtrees, each having at most $|V(\Lambda)|/2$ nodes \cite{jordan1869assemblages}; then, out of all the swap edges, we select a candidate edge $f$ minimizing the distance from $s$ to $c$ in $T_{e/f}$, i.e.,
\[
	f \in \displaystyle \arg\min_{(x',v') \in C(e)} \big\{d_T(s,x')+w(x',v')+d_{T}(v',c)\big\};
\]
then, we compute a \emph{critical node} $z$ for the selected swap edge $f$, i.e.,
\[
	z \in \arg\max_{z' \in V(T_b)}\sigma_{G-e}\big(T_{e/f},z'\big).
\]
We now select a suitable subtree $\Lambda'$ of the forest $F$,
and we pass to the selection of the next candidate BSE by recursing on $\Lambda'$, until $|V(\Lambda')|=1$. More precisely, $\Lambda'$ is the first tree of $F$ containing the first vertex of $V(\Lambda)$ that is encountered by following the path in $T$ from $z$ towards $c$ (see Figure~\ref{fig:centroid}).

Due to the property of the centroid, the number of recursions will be $O\big(\log |V(T_b)|\big)= O(\log n)$, as promised, each costing $O(n)$ time. Moreover, at least one of the candidate edges will be a BSE for $e$, and hence it suffices to choose the edge minimizing the maximum stretch among the corresponding $O(\log n)$ candidate edges.
This step is done within the recursive procedure by comparing the current candidate edge $f$ with the best candidate resulting from the nested recursive calls.

A more formal description of each phase is shown in Algorithm~\ref{alg:absems}. In the following we prove the correctness of our algorithm.

\begin{figure}[t]
	\centering
	\includegraphics[scale=1.2]{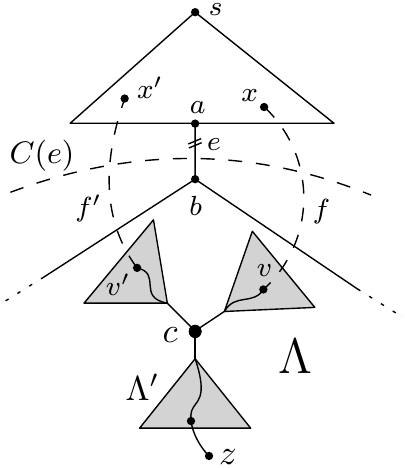}
	\caption{The situation illustrated in Lemma~\ref{lemma:centroid}. The subtree $\Lambda$ is represented by the three gray
triangles along with the vertex c. $f=(x,v)$ is the candidate swap edge for $e$ that minimizes $d_T(s,x)+w(f)+ d_T(v,c)$, and $z$ is its corresponding critical node. The algorithm will compute the next candidate swap edge by recursing on $\Lambda'$.}
	\label{fig:centroid}
\end{figure}

\begin{lemma}
\label{lemma:centroid}
Let $e=(a,b)$ be a failing edge, and let $\Lambda$ be a subtree of $T_b$. Given a vertex $c \in V(\Lambda)$, let $f \in \arg\min_{ (x',v') \in C(e)} \big\{ d_T(s,v') + w(x',v') + d_T(v', c) \big\}$ and let $z$ be a critical node for $f$. Let $F$ be the forest obtained by removing the edges incident to $c$ from $\Lambda$, and let $\Lambda'$ be the tree of $F$ containing the first vertex of the path from $z$ to $c$ in $T$ that is also in $V(\Lambda)$. For any swap edge $f' \in C(e,V(\Lambda))$, if $\mu(f') < \mu(f)$ then $f' \in C(e, V(\Lambda'))$.
\end{lemma}
\begin{proof}
	Let $f=(x, v)$ and $f' = (x', v')$.
	We show that if $v' \in V(\Lambda) \setminus V(\Lambda')$ then $\mu(f') \ge \mu(f)$ (see also Figure~\ref{fig:centroid}). Indeed:
	\begin{align*}
		\mu(f') & \ge \sigma_{G-e}\big(T_{e/f'}, z\big) =
		\frac{ d_{T_{e/f'}}(s,z)}{d_{G-e}(s, z)}
		= \frac{ d_T(s,x') + w(f') + d_T(v',c) + d_T(c,z) }{d_{G-e}(s, z)} \\
		& \ge \frac{ d_T(s,x) + w(f) + d_T(v,c) + d_T(c,z) }{d_{G-e}(s, z)}
		\ge \frac{ d_{T_e/f}(s,z) }{d_{G-e}(s, z)} = \sigma_{G-e}\big(T_{e/f}, z\big) \\
		& = \mu(f),
	\end{align*}
where we used the equality $d_T(v',z) = d_T(v',c) + d_T(c,z)$, which follows from the fact that the path from $v'$ to $z$ in $T$ must traverse $c$ as $v'$ and $z$ are in two different trees of $F$.
\qed
\end{proof}

\begin{algorithm}[t]{\caption{\absems{}($e, \Lambda$)}}
	\label{alg:absems}
	\SetKwInOut{Input}{Input}
	\SetKwInOut{Output}{Output}
	\Input{a failing edge $e=(a,b) \in E(T)$, a subtree $\Lambda$ of $T_b$.}
	\Output{an edge $f \in C(e)$. If $C\big(e, V(\Lambda)\big)$ contains a BSE for $e$, $f$ is a BSE for $e$.}

	\BlankLine
	
	$c \gets $ Centroid of $\Lambda$\;
	
	Let $f \in \displaystyle \arg\min_{ (x',v') \in C(e)} \big\{ d_T(s,x') + w(x',v') + d_T(v', c) \big\}$\;

	\BlankLine
	
	\lIf{ $|V(\Lambda)|=1$ }{\Return $f$}

	\BlankLine
    Let $z \in \arg \max_{z' \in V(T_b)}\sigma_{G-e}\big(T_{e/f},z'\big)$\;
	Let $F$ be the forest obtained by removing the edges incident to $c$ from $\Lambda$\;
	$y \gets $ first vertex along the path from $z$ towards $c$ in $T$ that is also in $V(\Lambda)$\;
	$\Lambda' \gets$ tree of $F$ containing $y$\;
	$f' \gets \absems{}(e, \Lambda')$\;
	
	\BlankLine
		
	\leIf{ $\mu(f') < \mu(f)$ }{\Return $f'$}{\Return $f$}
\end{algorithm}

\begin{lemma}
	If $C\big(e, V(\Lambda)\big)$ contains a BSE for $e$ then $\absems{}(e, \Lambda)$ returns a BSE for $e$.
	\label{lemma:absems_correctness}
\end{lemma}
\begin{proof}
	First of all notice that Algorithm~\ref{alg:absems} only returns edges in $C(e)$.

	We prove the claim by induction on $|V(\Lambda)|$.
	If $|V(\Lambda)|=1$ and $C\big(e, V(\Lambda)\big)$ contains a BSE $f^*$ for $e$, then let $f$ be the edge of $C(e)$ returned by Algorithm~\ref{alg:absems} and let $V(\Lambda)=\{c\}$. By choice of $f$, for every $v \in V(T_b)$,
$$d_{T_{e/f}}(s,v)\leq d_{T_{e/f}}(s,c)+d_T(c,v)=d_{T_{e/f^*}}(s,c)+d_T(c,v) = d_{T_{e/f^*}}(s,v),$$
from which we derive that $\mu(f) = \mu(f^*)$, and the claim follows.
	
	If $|V(\Lambda)| > 1$ and $C\big(e, V(\Lambda)\big)$ contains a BSE for $e$, we distinguish two cases depending on whether the edge $f$ computed by Algorithm~\ref{alg:absems} is a BSE for $e$ or not. If that is the case, then $\mu(f) \le \mu(f'') \, \forall f'' \in C(e)$ and the algorithm correctly returns $f$.
	Otherwise, by Lemma~\ref{lemma:centroid}, any edge $f' \in C(e, V(\Lambda))$ such that $\mu(f') < \mu(f)$ must belong to $C\big(e, V(\Lambda')\big)$.
	It follows that $\Lambda'$ contains a BSE for $e$ and since $1 \le |V(\Lambda')| < |V(\Lambda)|$ we have, by inductive hypothesis, that the edge $f'$ returned by $\absems{}(G, e, \Lambda')$ is a BSE for $e$. Clearly $\mu(f') < \mu(f)$ and hence Algorithm~\ref{alg:absems} correctly returns $f'$.
\qed
\end{proof}

Since each invocation of Algorithm~\ref{alg:absems} requires $O(n)$ time, Lemma~\ref{lemma:absems_correctness} together with the previous discussions allows us to state the main theorem of this section:

\begin{theorem}
	There exists an algorithm that solves the \absems problem in $O(mn + n^2 \log n)$ time.
\end{theorem}

\section{An Algorithm for \abseas}
In this section we show how the \abseas problem can be solved efficiently in $O(mn \log \alpha (m,n))$ time.
Our approach first of all, in a preprocessing phase, computes in $O(m \, n \log \alpha(m,n))$ time all the replacement shortest paths from the source after the failure of every edge of $T$ \cite{GP07}. Then, the algorithm will run through $n-1$ phases, each returning in $O(m)$ time a BSE for a failing edge of $T$, as described in the following. Thus, the overall time complexity will be dominated by the preprocessing step.

Let us fix $e=(a,b)\in E(T)$ as the failing edge of $T$. The idea is to show that, after a $O(n)$ preprocessing time, we can compute the avg-stretch $\lambda(f)$ of any $f$ in constant time. This immediately implies that we can compute a BSE for $e$ by looking at all $O(m)$ swap edges for $e$.

Let $U=V(T_b)$ and let $y$ be a node in $U$, we define:
\[ M(y)=\sum_{v \in U} \frac{d_T(y,v)}{d_{G-e}(s,v)} \] and \[Q = \sum_{v \in U} \frac{1}{d_{G-e}(s,v)}. \]
Let $f=(x,y)$ be a candidate swap edge incident in $y \in U$. The avg-stretch of $f$ can be rewritten as:

\[\lambda(f)=\sum_{v \in U}\frac{d_T(s,x)+w(f)+d_T(y,v)}{d_{G-e}(s,v)} = \big(d_T(s,x)+w(f)\big)Q + M(y).\]

Hence, the avg-stretch of $f$ can be computed in $O(1)$ time, once $Q$ and $M(y)$ are available in constant time. Observe that $Q$ does not depend on $y$ and can be computed in $O(n)$ time. The rest of this section is devoted to show how to compute $M(y)$ for every $y \in U$ in $O(n)$ overall time.

\subsubsection*{Computing $M(y)$ for all $y \in U.$}
Let $y$ e $y'$ be two nodes in $U$ such that $y$ is a child of $y'$ in $T$. Moreover, let $U_y=V(T_y)$, and let $Q_y = \sum_{v \in U_y} \frac{1}{d_{G-e}(s,y)}$. Hence, we can rewrite $M(y)$ and $M(y')$ as follows:

\[M(y) = \sum_{v \in U_y} \frac{d_T(y,v)}{d_{G-e}(s,v)} + \sum_{v \in U-U_y} \frac{w(y,y')+d_T(y',v)}{d_{G-e}(s,v)}\]
and
\[M(y') = \sum_{v \in U_y} \frac{w(y,y')+d_T(y,v)}{d_{G-e}(s,v)} + \sum_{v \in U-U_y} \frac{d_T(y',v)}{d_{G-e}(s,v)}.\]
Therefore, we have:
\begin{equation}\label{eq:propsum}
M(y) = M(y') + w(y,y') \big(- Q_y + (Q - Q_y)\big) = M(y') + w(y,y') \big(Q - 2Q_y\big).
\end{equation}
\noindent The above equation implies that $M(y)$ can be computed in $O(1)$ time, once we have computed $M(y')$, $Q$ and $Q_y$. As a consequence, we can compute all the $M(y)$'s as follows. First, we compute $Q_y$ for every $y\in D$ in $O(n)$ overall time by means of a postorder visit of $T_b$. Notice also that $Q=Q_b$. Then, we compute $M(b)$ explicitly in $O(n)$ time. Finally, we compute all the other $M(y)$'s by performing a preorder visit of $T_b$. When we visit a node $y$, we compute $M(y)$ in constant time using (\ref{eq:propsum}). Thus, the visit will take $O(n)$ time. We have proved the following:
\begin{theorem}
There exists an algorithm that solves the \abseas problem in $O(mn \log \alpha (m,n))$ time.
\end{theorem}

\section{An approximate solution for \absems}
In this section we show that for the max-stretch measure we can compute in an almost linear $O(m \log \alpha(m,n))$ time, a set of \emph{good} swap edges (GSE), each of which guarantees a relative approximation factor on the maximum stretch of $3/2$ (tight), as opposed to that provided by the corresponding BSE. Moreover, as shown in the next section, each GSE still guarantees an absolute maximum stretch factor w.r.t. the paths emanating from the source (in the surviving graph) equal to 3 (tight).

\begin{lemma} \label{thm:apx-absems}
Let $e$ be a failing edge in $T$, let
\[g=(x,y) \in \arg \min_{(x',v') \in C(e)}\big\{d_T(s,x')+w(x',v')\big\},\]
\noindent and, finally, let $f=(x',y')$ be a best swap edge for $e$ w.r.t. \absems. Then, $\mu(g)/\mu(f) \leq 3/2$.
\end{lemma}
\begin{proof}
Let $z$ be the critical node for the good swap edge $g$, and let $t$ (resp., $t'$) denote the \emph{least common ancestor} in $T$ between $y'$ and $z$ (resp., $y'$ and $y$). Let $D=d_T(s,x)+w(x,y)=d_{G-e}(s,y)$. By choice of $g$, it holds that $d_{G-e}(s,z) \geq D$ and $d_{G-e}(s,y' ) \geq D$. We divide the proof into the following two cases, as depicted in Figure~\ref{fig:quality_max}: either (1) $t$ is an ancestor of $t'$ in $T$, or (2) $t'$ is an ancestor of $t$ in $T$. Let $A,B,C$ denote the distance in $T$ between $y$ and $t'$, $t'$ and $t$, $t$ and $z$, respectively.
\begin{figure}[t]
	\centering
	\includegraphics[scale=1]{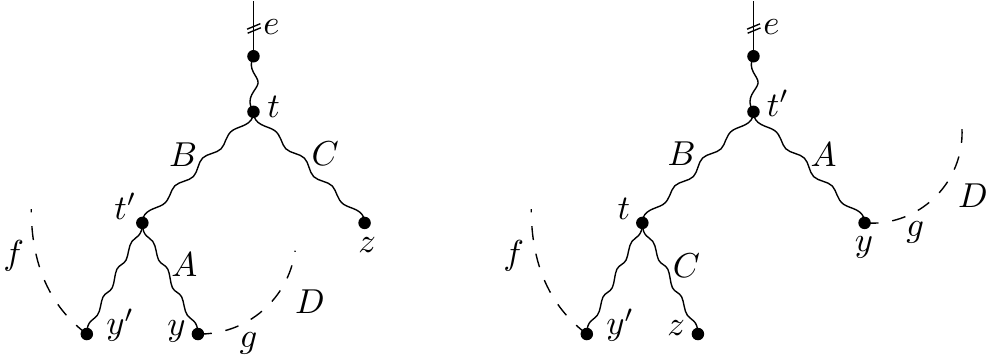} \caption{The figure shows the two cases of the analysis, on the left $t$ is an ancestor of $t^\prime$, while on the right the opposite holds. The splines denote a path, while the straight lines represent a single edge.} \label{fig:quality_max}
\end{figure}
\subsubsection*{Case 1}
Since $t$ is an ancestor of $t'$ (left side of Figure~\ref{fig:quality_max}), we have that $d_{T_{e/f}}(s,y) \geq D+A$ and we can write:
\[
	\sigma_{G-e}(T_{e/f},y) \geq \frac{D+A}{d_{G-e}(s,y)} = \frac{D+A}{D} \geq \frac{D+A}{d_{G-e}(s,z)},
\]

\noindent
and similarly $\sigma_{G-e}(T_{e/f},z) \geq \frac{D+B+C}{d_{G-e}(s,z)}$.
Moreover, by the definition of $\mu(\cdot)$ we have that $\mu(f) \geq \max\{\sigma_{G-e}(T_{e/f},y), \sigma_{G-e}(T_{e/f},z)\}$. The previous inequalities together imply:
\begin{align}\label{eq:case1}
	\frac{\mu(g)}{\mu(f)} \leq \frac{\sigma_{G-e}(T_{e/g},z)}{\max\{\sigma_{G-e}(T_{e/f},y), \sigma_{G-e}(T_{e/f},z)\}} \leq \frac{A+B+C+D}{D+\max\left\{ A, B+C \right\}.}
\end{align}

Now we divide the proof into two subcases, depending on whether $B+C \geq A$ or $B+C < A$. Observe that $D \geq d_G(s,y) \geq A$.
If $B+C \geq A$, then \eqref{eq:case1} becomes:
\[
	\frac{\mu(g)}{\mu(f)} \leq \frac{A+B+C+D}{B+C+D} = 1 + \frac{A}{B+C+D} \leq 1 + \frac{A}{2A} = \frac{3}{2},
\]
otherwise, if $B+C < A$, then  \eqref{eq:case1} becomes:
\[
	\frac{\mu(g)}{\mu(f)} \leq \frac{A+B+C+D}{A+D} < \frac{2A+D}{A+D} = 1 + \frac{A}{A+D} \leq 1 + \frac{A}{2A} = \frac{3}{2}.
\]

\subsubsection*{Case 2}
Assume now that $t'$ is an ancestor of $t$ (right side of Figure~\ref{fig:quality_max}). Since
 \[\mu(f) \geq \sigma_{G-e}(T_{e/f},y) \geq \frac{d_{G-e}(s,y')+A+B}{d_{G-e}(s,y)}=\frac{d_{G-e}(s,y')+A+B}{D},\]
we have that:
\begin{align*}
	\frac{\mu(g)}{\mu(f)}& \leq \frac{A+B+C+D}{d_{G-e}(s,z)} \cdot \frac{D}{d_{G-e}(s,y')+A+B} \\
	&\leq \frac{A+B+C+D}{d_{G-e}(s,z)} \cdot \frac{D}{A+B+D}
\end{align*}
and since $d_{G-e}(s,z) \geq d_G(s,z) \geq C$, and recalling that $d_{G-e}(s,z) \geq D$, we have:
\begin{equation}
\label{eq:case2}
	\frac{\mu(g)}{\mu(f)} \leq \frac{A+B+C+D}{A+B+D} \cdot \frac{D}{\max \left\{ C,D \right\}}
	= \left( 1 + \frac{C}{A+B+D} \right) \cdot \frac{D}{\max \left\{ C,D \right\}}.
\end{equation}
Moreover, notice that also the following holds:
\begin{align}
\nonumber \frac{\mu(g)}{\mu(f)} & \leq \frac{\mu(g)}{\sigma_{G-e}(T_{e/f},z)} \leq \frac{A+B+C+D}{d_{G-e}(s,z)} \cdot \frac{d_{G-e}(s,z)}{d_{G-e}(s,y') + d_T(y',t)+C} \\
 &\leq \frac{A+B+C+D}{C+D} = 1+ \frac{A+B}{C+D}. \label{eq:boh}
\end{align}
We divide the proof into the following two subcases, depending on whether $D \geq C$ or $D < C$.
In the first subcase, i.e., $D \geq C$, we have that \eqref{eq:case2} becomes $\frac{\mu(g)}{\mu(f)} \le 1 + \frac{C}{A+B+D}$,
and hence, by combining this inequality with (\ref{eq:boh}), we obtain:
\begin{align*}
\frac{\mu(g)}{\mu(f)} & \leq 1+ \min \left\{ \frac{C}{A+B+D}, \frac{A+B}{C+D}\right\} \\
&\leq 1+ \min \left\{ \frac{C}{A+B+C}, \frac{A+B}{2C}\right\} \leq 1+\frac{1}{2}=\frac{3}{2}.
\end{align*}
In the second subcase, i.e., $D < C$, \eqref{eq:case2} becomes:
\begin{equation}
\label{eq:subcase2.2}
\frac{\mu(g)}{\mu(f)}  \leq \left( 1 + \frac{C}{A+B+D} \right) \cdot \frac{D}{C} \leq \frac{D}{C} + \frac{D}{A+B+D} < 1 + \frac{D}{A+B+D},
\end{equation}
and hence, by combining \eqref{eq:subcase2.2} and \eqref{eq:boh}, we have that:
\begin{align*}
 \nonumber \frac{\mu(g)}{\mu(f)}  &\leq 1+ \min \left\{ \frac{D}{A+B+D}, \frac{A+B}{C+D}\right\} \\
 &\leq 1+ \min \left\{ \frac{D}{A+B+D}, \frac{A+B}{2D}\right\} \leq 1+\frac{1}{2}=\frac{3}{2},
\end{align*}
from which the claim follows.
\qed
\end{proof}

Given the result of Lemma~\ref{thm:apx-absems}, we can derive an efficient algorithm to compute all the GSE for \absems. More precisely, in \cite{NPW03} it was shown how to find them in $O(m \, \alpha(m,n))$ time. Essentially, the approach used in \cite{NPW03} was based on a reduction to the \emph{SPT sensitivity analysis} problem \cite{Tar82}. However, in \cite{Pet05} it was proposed a faster solution to such a problem, running in $O(m \, \log \alpha(m,n))$ time. Thus, we can provide the following

\begin{theorem}
There exists a $3/2$-approximation algorithm that solves the \absems problem in $O(m \log \alpha (m,n))$ time.
\end{theorem}
We conclude this section with a tight example which shows that the analysis provided in Lemma~\ref{thm:apx-absems} is tight (see Figure~\ref{fig:tight_example_GSE_max_stretch}).

\begin{figure}[t]
	\centering
	\includegraphics[scale=1.1]{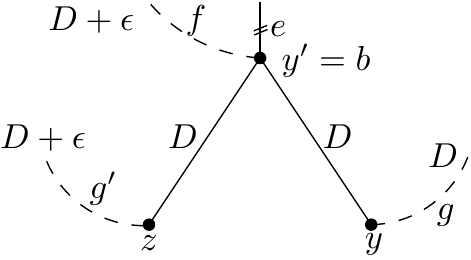}
	\caption{A tight example showing that the quality of the good swap edge $g$ computed by the algorithm is a factor of $3/2$ away from the qualify of a best swap edge $f$. In the picture, it is assumed that the distance from $s$ to $b$ is equal to 0, while the three dashed edges are assumed to be incident to the source, and the labels correspond to their weight. Then, $\mu(f) = \sigma_{G-e}(T_{e/f},y) =\frac{2D+\epsilon}{D} \simeq 2$, while  $\mu(g) = \sigma_{G-e}(T_{e/g},z) = \frac{3D}{D+\epsilon}\simeq 3$, for small values of $\epsilon$.}
	\label{fig:tight_example_GSE_max_stretch}
\end{figure}

\section{Quality analysis}
As for previous studies on swap edges, it is interesting now to see how the tree obtained from swapping
a failing edge $e=(a,b)$ with its BSE $f$ compares with a true SPT of $G-e$. According to our swap criteria, we will then analyze the lower and upper bounds of the max- and avg-stretch of $f$, i.e., $\mu(f)$ and $\lambda(f)$, respectively.

As already observed in the introduction, it is well-known \cite{NPW03} that for the swap edge, say $g$, which belongs to the shortest path in $G-e$ between $s$ and the root of the detached subtree $T_b$, we have that for any $v \in V(T_b)$, $\sigma_{G-e}(T_{e/g},v) \leq 3$. This immediately implies that $\mu(g),\lambda(g) \leq 3$, namely $\mu(f),\lambda(f) \leq 3$. These bounds happen to be tight, as shown in Figure \ref{fig:tight}.

\begin{figure}[t]
\centering
	\includegraphics[scale=.86]{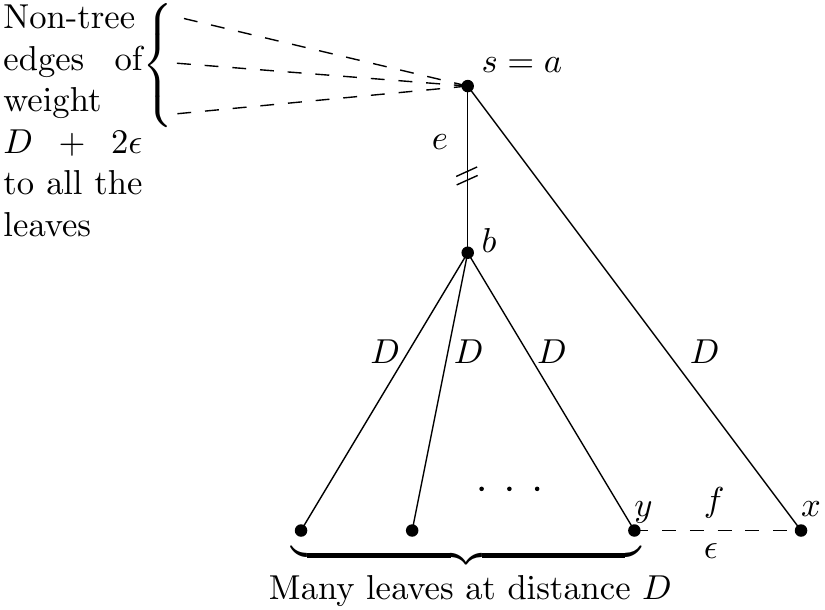}
\caption{Tight ratios for $\mu(f)$ and $\lambda(f)$.
In the picture, the SPT $T$ (solid edges) along with the removed edge $e=(a,b)$ of weight $0$; non-tree edges
are dashed and the best swap edge (for both the maximum and the average stretch) is easily seen to be $f=(x,y)$, of weight $\epsilon >0$.
Then, we have that $\mu(f)=\frac{3D+\epsilon}{D+2 \epsilon}$, which tends to $3$ for small values of $\epsilon$, while $\lambda(f)$ tends to $3$ as well, as soon as the number of leaves grows. Notice that $f$ is also a good swap edge for $e$.}
\label{fig:tight}
\end{figure}

Let us now analyze the lower and upper bounds of the max-stretch of a \emph{good} swap edge $g$, i.e., $\mu(g)$, as defined in the previous section.
First of all, once again it was proven in \cite{NPW01} that for any $v \in V(T_b)$, $\sigma_{G-e}(T_{e/g},v) \leq 3$, which implies that $\mu(g) \leq 3$. Moreover, the example shown in Figure \ref{fig:tight} can be used to verify that this bound is tight.

\section{Conclusions}
In this paper we have studied two natural SPT swap problems, aiming to minimize, after the failure of any edge of the given SPT, either the maximum or the average stretch factor induced by a swap edge. We have first proposed two efficient algorithms to solve both problems. Then, aiming to the design of faster algorithms, we developed for the maximum-stretch measure an almost linear algorithm guaranteeing a $3/2$-approximation w.r.t. the optimum.

Concerning future research directions, the most important open problem remains that of finding a linear-size edge-fault-tolerant SPT with a (maximum) stretch factor w.r.t. the root better than 3, or to prove that this is unfeasible. Another interesting open problem is that of improving the running time of our exact solutions. Notice that both our exact algorithms pass through the computation of all the post-failure single-source distances, and if we could avoid that we would get faster solutions. At a first glance, this sounds very hard, since the stretches are heavily dependant on post-failure distances, but, at least in principle, one could exploit some monotonicity property among swap edges that could allow to skip such a bottleneck.
Besides that, it would be nice to design a fast approximation algorithm for the average-stretch measure. Apparently, in this case it is not easy to adopt an approach based on good swap edges as for the maximum-stretch case, since swap edges optimizing other reasonable swap criteria (e.g., minimizing the distance towards the root of the detached subtree, or minimizing the distance towards a detached node) are easily seen to produce an approximation ratio of 3 as opposed to a BSE. A candidate solution may be that of selecting a BSE w.r.t. the sum-of-distances criterium, which can be solved in almost linear time \cite{DP07}, but for which we are currently unable to provide a corresponding comparative analysis.

Finally, we mention that a concrete task which will be pursued is that of conducting an extensive experimental analysis of the true performances of our algorithms, to check whether for real-world instances the obtained stretches are sensibly better or not w.r.t. the theoretical bounds.

\bibliographystyle{plain}
\bibliography{bib}
\end{document}